\documentclass[envcountsect,envcountsame,runningheads]{llncs}
\usepackage[utf8]{inputenc}

\usepackage{amsfonts}
\usepackage{amssymb}
\usepackage{amsmath}
\usepackage{graphicx}
\usepackage{color}
\usepackage{xspace}

\newcommand{\ignore}[1]{}

\newenvironment{proofof}[1]{\bigskip \noindent {\bf Proof of #1:}
}{\qed\par\vskip 4mm\par}

\def\@begintheorem#1#2{\par\bgroup{\sc #1\ #2. }\it\ignorespaces}
\def\@opargbegintheorem#1#2#3{\par\bgroup{\sc #1\ #2\ (#3). } \it\ignorespaces}
\def\@endtheorem{\egroup}
\newcommand{\bt}[1]{\begin{theorem}\label{#1}}
\newcommand{\bc}[1]{\begin{corollary}\label{#1}}
\newcommand{\bl}[1]{\begin{lemma}\label{#1}}
\newcommand{\be}[1]{\begin{example}\label{#1}}
\newcommand{\bp}[1]{\begin{proposition}\label{#1}}
\newcommand{\ba}[1]{\begin{algorithm}\rm\label{#1}}
\newcommand{\bd}[1]{\begin{definition}\rm\label{#1}}{\normalsize }
\newcommand{\bpr}{\begin{proof}}
\newcommand{\et}{\end{theorem}}
\newcommand{\ec}{\end{corollary}}
\newcommand{\el}{\end{lemma}}
\newcommand{\ee}{\end{example}}
\newcommand{\ep}{\end{proposition}}
\newcommand{\ed}{\end{definition}}
\newcommand{\epr}{\qed\end{proof}}

\def\R{\mathbb{R}}
\def\Z{\mathbb{Z}}
\def\N{\mathbb{N}}
\def\CC{\mathcal{C}}
\def\DD{\mathcal{D}}
\def\HH{\mathcal{H}}
\def\FF{\mathcal{F}}
\def\Oh{\mathcal{O}}

\def\la{\overline}
\def\vc#1{\mathbf{#1}}
\def\vcl{\mbox{\boldmath$\ell$}}
\def\bfx{\mathbf{x}}
\def\cc{{\overline{\vc c}}}
\def\xx{{\overline{\vc x}}}

\newcommand{\conv}{\mathop{\mathrm{conv}}}

\newcommand{\FPT}{{\sf FPT}\xspace}
\newcommand{\XP}{{\sf XP}\xspace}
\newcommand{\NP}{{\sf NP}\xspace}
\newcommand{\PP}{{\sf P}\xspace}
\newcommand{\W}[1]{{\sf W}[#1]\xspace}

\newcommand{\pedge}{\textsl{edge}}
\def\prebox#1{\mathop{\mbox{\rm #1}}}

\usepackage{ifthen}
\newcounter{Accumulate} 
\ifthenelse{\value{Accumulate} = 1}{
  \newwrite\accuwrite \immediate\openout\accuwrite=\jobname.acc
}{}
\usepackage{environ}
\makeatletter
\newenvironment{accumulate}{\Collect@Body\accuPrint}{}
\makeatother
\newcommand{\accuPrint}[1]{
 \ifthenelse{\value{Accumulate} = 0}{%
      #1
  }
  {
    \newtoks\prxxxm
    \prxxxm{#1}
    \immediate\write\accuwrite{\the\prxxxm}
  }
}
\newcommand{\ifaccumulating}[1]{%
  \ifthenelse{\value{Accumulate} = 1}{%
    #1%
  }{}%
}
\newcommand{\ifnoaccumulating}[1]{%
  \ifthenelse{\value{Accumulate} = 0}{%
    #1%
  }{}%
}
\newcommand{\accuprint}{%
  \ifthenelse{\value{Accumulate} = 1}{
    \immediate\closeout\accuwrite
    \input{\jobname.acc} %
  }{}
}
\ifthenelse{\value{Accumulate} = 0}{\let\APXmark\relax}
	{\def\APXmark{{\bf$\!\!$(*)~}}}

\begin{document}

\title{\bf Parameterized Shifted Combinatorial Optimization}

\author{
Jakub Gajarsk\'y%
  \thanks{Current affiliation: Technical University Berlin,
	e-mail jakub.gajarsky@tu-berlin.de
	J.~Gajarsk\'y's research was partially supported by the 
	European Research Council under the European Union's Horizon
	2020 research and innovation programme (ERC Consolidator Grant
	DISTRUCT, grant agreement No 648527).}\inst{1}
\and
Petr Hlin\v en\'y%
  \thanks{P.~Hlin\v en\'y, and partially J.~Gajarsk\'y, were
	supported by the the research centre Institute
   	for Theoretical Computer Science (CE-ITI), project P202/12/G061
	of the Czech Science Foundation.}\inst{1}
\and
Martin Kouteck\'y%
  \thanks{M.~Koutecký was partially supported by the project 17-09142S of
	the Czech Science Foundation.}\inst{2}
\and
Shmuel Onn%
  \thanks{Shmuel Onn was partially supported by the Dresner Chair at the
	Technion.}\inst{3}
}

\institute{\small Masaryk University, Brno, Czech Republic.
        Email: \small\texttt{\{gajarsky,hlineny\}@fi.muni.cz}
        \and
        Charles University, Prague.
		Email: \small\texttt{koutecky@kam.mff.cuni.cz}
		\and
		Technion - Israel Institute of Technology, Haifa, Israel.
		Email: \small\texttt{onn@ie.technion.ac.il}
        }


\maketitle

\begin{abstract}
%
{\em Shifted combinatorial optimization} is a new nonlinear optimization
framework which is a broad extension of standard combinatorial optimization, 
involving the choice of several feasible solutions at a time. 
This framework captures well studied and diverse problems
ranging from so-called vulnerability problems to sharing and partitioning
problems.
In particular, every standard combinatorial optimization problem has its
shifted counterpart, which is typically much harder. 
Already with explicitly given input set the shifted problem may be \NP-hard.
In this article we initiate a study of the parameterized complexity
of this framework. 
First we show that shifting over an explicitly given set with
its cardinality as the parameter may be in \XP, \FPT or \PP, 
depending on the objective function.
Second, we study the shifted problem over sets definable in MSO logic
(which includes, e.g., the well known MSO partitioning problems).
Our main results here are that shifted combinatorial optimization over MSO
definable sets is in \XP with respect to the MSO formula and the treewidth 
(or more generally clique-width) of the input graph, 
and is \W1-hard even under further severe restrictions.


\end{abstract}

\ifthenelse{\value{Accumulate} = 0}{\vfill}{}

\section{Introduction}

The following optimization problem has been studied extensively in the literature.

\vskip.2cm\noindent{\bf (Standard) Combinatorial Optimization.}
Given $S\subseteq\{0,1\}^n$ and $\vc w\in\Z^n$, solve
\vspace*{-1ex}\begin{equation}\label{standard}
\max\{\vc w\vc s\mid \vc s\in S\}\ .
\end{equation}
The complexity of the problem depends on $\vc w$ and the type and presentation of $S$.
Often, $S$ is the set of indicating (characteristic) vectors of members of a family of subsets over a ground
set $[n]:=\{1,\dots,n\}$, such as the family of $s-t$ dipaths in a digraph with
$n$ arcs, the set of perfect matchings in a bipartite or arbitrary graph with $n$ edges,
or the set of bases in a matroid over $[n]$ given by an independence oracle.

\smallskip
Partly motivated by vulnerability problems studied recently in the
literature (see a brief discussion below), 
in this article we study a broad nonlinear extension of CO, 
in which the optimization is over $r$ choices of elements of $S$ and which is
defined as follows.
For a set $S\subseteq\R^n$, let $S^r$ denote the set of $n\times r$ matrices having each column in $S$,
$$S^r\ :=\ \{\vc x\in\R^{n\times r}\mid \vc x^k\in S\,,\ k=1,\dots,r \}\ .$$
Call $\vc x,\vc y\in\R^{n\times r}$ equivalent and write $\vc x\sim\vc y$ if each row of
$\vc x$ is a permutation of the corresponding row of $\vc y$. 
The {\em shift} of $\vc x\in\R^{n\times r}$
is the unique matrix $\xx\in\R^{n\times r}$ satisfying $\xx\sim\vc x$ and
$\xx^1\geq \cdots\geq \xx^r$, that is, the unique matrix equivalent to $\vc x$ with each
row nonincreasing. 
Our nonlinear optimization problem follows:

\vskip.2cm\noindent{\bf Shifted Combinatorial Optimization (SCO).}
Given $S\subseteq\{0,1\}^n$ and $\vc c\in\Z^{n\times r}$, solve
\vspace*{-1ex}\begin{equation}\label{shift}
\max\{\vc c\xx\mid\vc x\in S^r\}\ .
\end{equation}
(Here $\vc c\xx$ is used to denote the ordinary scalar product of the
vectors $\vc c$ and~$\xx$.)

This problem easily captures many classical fundamental problems. 
%
For example, given a graph $G=(V,E)$ with $n$ vertices, let
$S:=\{N[v] \mid v \in V\} \subseteq \{0,1\}^n$, where $N[v]$ is the characteristic vector of the closed neighborhood of $v$.
Choose an integer parameter $r$ and
let $c^1_i:= 1$ for all $i$ and $c^j_i:= 0$ for all $i$ and all $j\geq 2$. 
Then the optimal objective function value of \eqref{shift} is $n$ 
if and only if we can select a set $D$ of $r$ vertices in $G$ such that every 
vertex belongs to the closed neighborhood of at least one of the selected vertices, 
that is, when $D$ is a dominating set of $G$.
Likewise, one can formulate the vertex cover and independent set problems in
a similar way.

One specific motivation for the SCO problem is as follows.
Suppose $S$ is the set of indicators of members of a family over $[n]$. 
A feasible solution $\vc x\in S^r$ then represents a choice of $r$ members
of the given family such that the $k$-th column $\vc x^k$ is the indicator 
of the $k$-th member. Call element $i$ in the
ground set {\em $k$-vulnerable} in $\xx$ if it is used by at least $k$ 
of the members represented by $\vc x$, that is, if the $i$-th row $\vc x_i$ of $\vc x$ has at least $k$ ones.
It is easy to see that the $k$-th column $\xx^k$ of the shift of $\vc x$ is precisely the indicator of the
set of $k$-vulnerable elements in $\vc x$. So the shifted optimization problem is to maximize
$$\vc c\xx\ =\ \sum\{c_i^k\mid \mbox{$i$ is $k$-vulnerable in $\vc x$},\ i=1,\dots,n\,,\ k=1,\dots,r\}\ .$$
Minimizing the numbers of $k$-vulnerable elements in $\vc x$ may be
beneficial for survival of some family members under various
attacks to vulnerable elements by an adversary, see e.g. \cite{AssadiENYZ:14,OmranSZ:2013} for more details.
For example, to minimize the number of $k$-vulnerable elements for some $k$, we set $c_i^k:=-1$
for all $i$ and $c_i^j:=0$ for all $i$ and all $j\neq k$. To {\em lexicographically} minimize
the numbers of $r$-vulnerable elements, then of $(r-1)$-vulnerable elements, and so on,
till that of $1$-vulnerable elements, we can set $c_i^k:=-(n+1)^{k-1}$ for all $i$, $k$. 

As another natural example, consider $\vc c$ with $c_i^1 := 1$ and 
$c_i^j := -1$ for $1 < j \leq r$. 
Then $\vc c \xx = n$ if and only if the columns of $\vc x$ indicate a {\em partition} of $S$. 
This formulation hence allows us to optimize over partitions of the ground
set (see Section~\ref{sec:xp_sco}).
Or, consider $\vc c$ with $\vc c_i = (1,\dots, 1, -1, \dots, -1)$
of length $a>0$ with $b\leq a$ ones, and let $S$ be the family of independent sets of a graph $G$. 
Then $\max \vc c\xx$ relates to {\em fractional coloring of~$G$};
it holds $\max \vc c\xx = bn$ if and only if $G$ has a coloring by $a$
colors in total such that every vertex receives $b$ distinct colors --
this is the so-called \textsc{$(a:b)$-coloring} problem.

\smallskip
The complexity of the shifted combinatorial optimization (SCO) problem 
depends on $\vc c$ and on the presentation of $S$, 
and is typically harder than the corresponding standard combinatorial optimization problem. 
Say, when $S$ is the set of perfect matchings in a graph,
the standard problem is polynomial time solvable, but the shifted problem is \NP-hard even
for $r=2$ and cubic graphs, as the optimal value of the above $2$-vulnerability problem is $0$
if and only if the graph is $3$-edge-colorable \cite{LevinO:2016}. 
The minimization of $2$-vulnerable arcs with $S$ the set of $s$--$t$ dipaths in a
digraph, also called the {\sc Minimum shared edges} problem, was recently shown
to be \NP-hard for $r$ variable in \cite{OmranSZ:2013}, polynomial time solvable for fixed $r$
in \cite{AssadiENYZ:14}, and fixed-parameter tractable with $r$ as a parameter in \cite{FluschnikKNS:2015}.

In the rest of this article we always assume that the number $r$ of choices is variable. 
Call a matrix $\vc c\in\Z^{n\times r}$ {\em shifted} if $\vc c=\cc$, 
that is, if its rows are nonincreasing.
In \cite{KaibelOS:2015} it was shown that when 
$S=\{\vc s\in\{0,1\}^n\mid \vc A\vc s=\vc b\}$ 
where $\vc A$ is a totally unimodular
matrix and $\vc b$ is an integer vector, the shifted problem with shifted
$\vc c$, and hence in particular the above lexicographic vulnerability problem, can be solved in polynomial time.
In particular this applies to the cases of $S$ the set of $s$--$t$ dipaths in a digraph
and $S$ the set of perfect matchings in a bipartite graph. In \cite{LevinO:2016} it was shown that the
shifted problem with shifted $\vc c$ is also solvable in polynomial time for $S$ the set of
bases of a matroid presented by an independence oracle (in particular, spanning trees in a graph),
and even for the intersection of matroids of certain type.

\vskip.2cm\noindent{\bf Main results and paper organization.}
In this article we continue on systematic study of shifted combinatorial optimization.
\ifthenelse{\value{Accumulate} = 0}{%
  The paper is organized as follows. 
  Preliminaries, including necessary basics of parameterized complexity 
  (\FPT, \XP and \W1-hardness) and of logic (MSO on graphs), are in Section~\ref{sec:preliminaries}. 
\par Then}{\par First},
in Section~\ref{sec:explicit}, we consider the case when the set $S$ is given explicitely.
While the standard problem is always trivial in such case,
the SCO problem can be \NP-hard for explicit set~$S$
(Proposition~\ref{hardness}).
Our main results on this case can be briefly summarized as follows:
\ifthenelse{\value{Accumulate} = 0}{}{\vspace*{-1ex}}%
\begin{itemize}
\item{\bf(Theorem~\ref{thm:explicit_sets}) } The shifted combinatorial optimization
problem, parameterized by $|S|=m$, is;
{(a)} for general $\vc c$ in the complexity class \XP and \W{1}-hard
w.r.t.~$m$,
{(b)} for shifted $\vc c$ in \FPT, and 
{(c)} for shifted $-\vc c$ in \PP.
\item{\bf(Theorem~\ref{shiftedt}) } The latter case {(c)} 
of shifted $-\vc c$ is in \PP even
for sets $S$ presented by a linear optimization oracle.
\end{itemize}
\ifthenelse{\value{Accumulate} = 0}{\par}{\vspace*{-0.5ex}}%
In Section~\ref{sec:xp_sco}, we study a more general framework of SCO
for the set $S$ definable in Monadic Second Order (MSO) logic.
This rich framework includes, for instance, the well-studied case of
so called MSO partitioning problems on graphs.
We prove the following statement which generalizes known results about MSO
partitioning:%
\ifthenelse{\value{Accumulate} = 0}{}{\vspace*{-1ex}}%
\begin{itemize}
\item{\bf(Theorem~\ref{thm:sco_xp}, Corollary~\ref{cor:sco_xp_cw}) } 
The shifted combinatorial optimization problem, for
(a) graphs of bounded treewidth and $S$ defined in MSO$_2$ logic, or
(b) graphs of bounded clique-width and $S$ defined in MSO$_1$ logic,
is in~\XP (parameterized by the width and the formula defining~$S$).
\end{itemize}
\ifthenelse{\value{Accumulate} = 0}{}{\vspace*{-0.5ex}}%
In the course of proving this statement we also provide a connection of shifted
optimization to separable optimization when
the corresponding polyhedron is decomposable and $0/1$
(Lemma~\ref{lem:idp_separable}).

To complement the previous tractability result, in
Section~\ref{sec:mso_part_hardness} we prove the following negative result
under much more restrictive parametrization.
\ifthenelse{\value{Accumulate} = 0}{}{\vspace*{-1ex}}%
\begin{itemize}
\item{\bf(Theorem~\ref{thm:part-hardinstance}) } 
There exists a fixed First Order formula $\phi$ such that the associated MSO$_1$
partitioning problem, and hence also the SCO problem with $S$ defined
by~$\phi$, are \W{1}-hard on graphs of bounded treedepth.
\end{itemize}

\ifthenelse{\value{Accumulate} = 1}{%
  We use standard terminology of graph theory, integer programming and
  parameterized complexity.
  Due to space restrictions, a short overview of basic terminology and
  the proofs of all our statements are moved to the Appendix.
  The statements with proofs presented in the Appendix are marked
  with~~\APXmark.
}{%
We conclude the paper by listing several interesting future 
research directions in Section~\ref{sec:conclusions}.
}

\begin{accumulate}
\ifthenelse{\value{Accumulate} = 0}
	{\section{Preliminaries}}
	{\subsection{Basic definitions}}
\label{sec:preliminaries}


We follow the standard terminology of graph theory and of linear and integer
programming.
Relevant special terminology is introduced in the corresponding sections of the paper.
Here we briefly recall basic terms of parameterized complexity
and of logic on graphs.

\smallskip
A {\em parameterized problem}~$Q$ is a subset of~$\Sigma^* \times \mathbb{N}_0$, 
where~$\Sigma$ is a finite alphabet.
A~parameterized problem~$Q$ is said to be \emph{fixed-parameter tractable}
if there is an algorithm that given~$(x,k) \in \Sigma \times \mathbb{N}_0$ decides whether~$(x,k)$
is a \textsc{yes}-instance of~$Q$ in time~$f(k) \cdot p(|x|)$ where~$f$ is some
computable function of~$k$ alone, $p$ is a polynomial and~$|x|$ is the size measure of
the input. 
The class of such problems is denoted by \FPT.
The class \XP is the class of parameterized problems that admit 
algorithms with a run-time of~$\Oh(|x|^{f(k)})$ for some computable
function~$f$, i.e.\ polynomial-time for every fixed value of~$k$.

Theory of parameterized complexity, see e.g.~\cite{DBLP:series/txcs/DowneyF13}, 
defines also complexity classes $\W t$ for $t\geq 1$,
where $\W t\subseteq \XP$ for all integers $t\geq 1$.
For instance, the $k$-independent set problem (with parameter $k$) is
complete for $\W1$.
Problems that are $\W1$-hard do not admit an FPT algorithm unless the
{\em Exponential Time Hypothesis} (ETH) fails, which is considered unlikely.

\smallskip
We now shortly introduce  \emph{monadic second order logic} (MSO) over graphs.
In {\em first-order logic} (FO) we have variables for the elements
($x,y,\ldots$), equality for variables, quantifiers $\forall,\exists$
ranging over vertices and vertex sets, and the standard Boolean connectives.
MSO is the extension of FO by quantification over sets ($X,Y,\dots$).
In this context a graph is considered as a {\em relational structure}; 
either with only the
adjacency relation on its vertex set (i.e., the relational vocabulary
consists of one predicate symbol $\prebox{edge}(x,y)\,$),
or as two-sorted structures with its vertex and edge sets and the incidence
relation between those (the vocabulary consists of $\prebox{inc}(x,e)$,
which can also be used to define adjacency).

In correspondence with the previous, graph MSO traditionally comes in two
flavours, MSO$_1$ and MSO$_2$, differing by the objects we are allowed to
quantify over:
in MSO$_1$ these are the vertices and vertex sets,
while in MSO$_2$ we can additionally quantify over edges and edge sets.
For example, the $3$-colorability property can be expressed in MSO$_1$ as follows:
\begin{eqnarray*}
 \exists X_1,X_2,X_3 &&\left[\,
 	\forall x \, (x\in X_1\vee x\in X_2\vee x\in X_3) \wedge \right.
\\ &&\bigwedge\nolimits_{i=1,2,3} \left. \!\!
   \forall x,y
 	\left(x\not\in X_i\vee y\not\in X_i\vee \neg\prebox{edge}(x,y)\right)
 \,\right]
\end{eqnarray*}
We briefly remark that MSO$_2$ can express properties which 
are not MSO$_1$ definable (e.g., Hamiltonicity).

\ifthenelse{\value{Accumulate} = 1}{\stepcounter{section}}{}
\end{accumulate}

\section{Sets Given Explicitly}\label{sec:explicit}
\begin{accumulate}
\ifthenelse{\value{Accumulate} = 1}{
\subsection{Additions to Section~\ref{sec:explicit}}
}{}\end{accumulate}

In this section we consider the shifted problem \eqref{shift} 
over an explicitly given set $S=\{\vc s^1,\dots,\vc s^m\}$.
We demonstrate that already this seemingly simple case is in fact nontrivial
and interesting.
First, notice that with $S\subseteq\{0,1\}^n$ given explicitly the problem 
is generally \NP-hard, which follows by the reduction from dominating set
which we gave in the introduction.  Moreover it follows from known lower
bounds on the dominating set problem that the brute-force algorithm which
tries all possible $r$-subsets of $S$ is likely close to optimal:

\bp{hardness}\APXmark
The SCO problem \eqref{shift} is \NP-hard for $0/1$ shifted matrices
$\vc c=\cc\in\{0,1\}^{n\times r}$ and explicitly given $0/1$ sets 
$S=\{\vc s^1,\dots,\vc s^m\}\subseteq\{0,1\}^n$. 
Moreover, unless the {\em Exponential Time Hypothesis} (ETH) fails, it cannot be solved in time $n^{o(r)}$.
\ep

\begin{accumulate}
\ifthenelse{\value{Accumulate} = 1}{%
  \proofof{Proposition~\ref{hardness}}
}{\bpr}
The \NP-complete \textit{dominating set} problem is to decide whether, given a graph $G=(V,E)$ there is a subset of vertices $D \subseteq V$ of size $r$ such that every vertex $v \in V$ is either in $D$, or has a neighbor in $D$.
Let $S:=\{N[v] \mid v \in V\} \subseteq \{0,1\}^n$, where $N[v]$ is the characteristic vector 
of the closed neighborhood of $v$, i.e.~including $v$ itself,
and let $c^1_i:= 1$ for all $i$ and $c^j_i:= 0$ for all $i$ and all $j\geq 2$. 
Then the optimal objective function value of \eqref{shift} is $n$ 
if and only if $G$ has a dominating set of size $r$.

Moreover, Chen et al.~\cite{ChenHKX:06} proved that unless ETH fails, there is no $n^{o(r)}$ algorithm solving the dominating set problem; thus, under the same assumption, there is no $m^{o(r)}$ algorithm solving SCO even when $c$ is $0/1$ and $c=\cc$.
\epr
\end{accumulate}

Note that the next results in this section concerning Shifted IP
apply to the more general situation 
in which $S$ may consist of arbitrary integer vectors, not necessarily $0/1$.
This is formulated as follows.

\vskip.2cm\noindent{\bf Shifted integer programming.}
Given $S\subseteq\Z^n$ and $\vc c\in\Z^{n\times r}$, similarly to
\eqref{shift}, solve
\vspace*{-1ex}\begin{equation}\label{shiftZ}
\max\{\vc c\xx\mid\vc x\in S^r\}
\,.\end{equation}

For $S=\{\vc s^1,\dots,\vc s^m\}$ and nonnegative integers $r_1,\dots,r_m$ 
with $\sum_{i=1}^m r_i=r$, let
$\vc x(r_1,\dots,r_m)$ be the matrix in $S^r$ with first $r_1$ columns equal to
$\vc s^1$, next $r_2$ columns equal to $\vc s^2$, and so on, 
with last $r_m$ columns equal to $\vc s^m$,
and define $f(r_1,\dots,r_m):=\vc c\xx(r_1,\dots,r_m)$.

We have got the following effective theorem in contrast with Proposition \ref{hardness}.
\bt{main}\APXmark\label{thm:explicit_sets}
The shifted integer programming problem \eqref{shiftZ} over an explicitly 
given set $S=\{\vc s^1,\dots,\vc s^m\}\subseteq\Z^n$
reduces to the following nonlinear integer programming problem over a simplex,
\begin{equation}\label{IP}
\max\left\{f(r_1,\dots,r_m)\left|~
  r_1,\dots,r_m\in\Z_+\,,\ \sum_{k=1}^m r_k=r\right\}\right..
\end{equation}
If $\vc c=\cc$ is shifted then $f$ is concave, and if $-\vc c$ is shifted then $f$ is convex.

Moreover, the following hold:
\begin{enumerate}\vspace*{-1ex}
\item
With $m$ parameter and $\vc c$ arbitrary, problem \eqref{shiftZ} is in \XP.
Furthermore, the problem is \W{1}-hard with parameter~$m$ even
for 0/1 sets~$S$.
\item
With $m$ parameter and $\vc c$ shifted, problem \eqref{shiftZ} is in \FPT.
\item
With $m$ variable and $-\vc c$ shifted, problem \eqref{shiftZ} is in \PP.
\end{enumerate}
\et

\begin{accumulate}
\ifthenelse{\value{Accumulate} = 1}{%
  \proofof{Theorem~\ref{main}}
}{\bpr}
Consider any $\vc x\in S^r$. For $k=1,\dots,m$ let
$r_k:=|\{j\mid \vc x^j=\vc s^k\}|$ be the number of columns of $\vc x$ equal to $\vc s^k$.
Then $\vc x\sim\vc x(r_1,\dots,r_m)$ so $\xx=\xx(r_1,\dots,r_m)$ and 
$\vc c\xx=f(r_1,\dots,r_m)$.
So an optimal solution $r_1,\dots,r_m$ to \eqref{IP} gives an optimal solution
$\vc x(r_1,\dots,r_m)$ to the shifted problem \eqref{shiftZ}, proving the first statement.

We next show that if $\vc c$ is shifted then $f$ is concave in the $r_k$.
Suppose first that $n=1$ so that $\vc c^1,\dots,\vc c^r$ and $\vc
s^1,\dots,\vc s^m$ are scalars.
For $k=1,\dots,m$, define functions $g_k(r_1,\dots,r_m):=\sum_{j=1}^k r_j$ which are linear
in $r_1,\dots,r_m$, and define a function $h$ by $h(0):=0$ and
$h(l):=\sum_{j=1}^l\vc c^j$
for $l=1,\dots,r$, which is concave since $\vc c^1\geq\cdots\geq\vc c^r$.

Let $\pi$ be a permutation of $[m]=\{1,\dots,m\}$ such that
$\vc s^{\pi(1)}\geq\cdots\geq\vc s^{\pi(m)}$. 
Consider any $r_1,\dots,r_m$ feasible in \eqref{IP}
and let $\vc x:=\vc x(r_1,\dots,r_m)$. Note that $\xx$ is the row vector with first $r_{\pi(1)}$
entries equal to $\vc s^{\pi(1)}$, next $r_{\pi(2)}$ entries equal to $\vc s^{\pi(2)}$,
and so on, with last $r_{\pi(m)}$ entries equal to $\vc s^{\pi(m)}$.
Let $g_k:=g_k(r_{\pi(1)},\dots,r_{\pi(m)})$ for $k=1,\dots,m$ and
$t^k:=\vc s^{\pi(k)}-\vc s^{\pi(k+1)}\geq 0$ for $k=1,\dots,m-1$. Then we have that
\begin{eqnarray}
\nonumber
  f(r_1,\dots,r_m) &=& \vc c\xx \\
\nonumber
  &=& \vc s^{\pi(1)}h(g_1)+\vc s^{\pi(2)}\left(h(g_2)-h(g_1)\right)+
  \cdots+\vc s^{\pi(m)}\left(h(g_m)-h(g_{m-1})\right)\\
\nonumber
  &=& \vc t^1h(g_1)+\vc t^2h(g_2)+\cdots+\vc t^{m-1}h(g_{m-1})+\vc s^{\pi(m)}h(g_m) \\
  &=& \sum_{k=1}^{m-1} \vc t^k h\!\left(g_k(r_{\pi(1)},\dots,r_{\pi(m)})\right)
      + \vc s^{\pi(m)}\sum_{j=1}^r \vc c^j \ .
\end{eqnarray}
Now, $g_k$ are linear functions of $r_k$, and $h$ is concave, and so each composition
$h(g_k(r_{\pi(1)},\dots,r_{\pi(m)}))$ is also concave. So $f(r_1,\dots,r_m)$, which is a constant
plus a nonnegative combination of concave functions, is a concave function of the $r_k$.

We continue with general $n$. Consider any $r_1,\dots,r_m$ which are feasible in \eqref{IP}
and let $\vc x:=\vc x(r_1,\dots,r_m)$. For each $i=1,\dots,n$ proceed as follows.
Let $f_i(r_1,\dots,r_m):=\vc c_i\xx_i$ with $\vc c_i$ the $i$-th row of $\vc c$ and $\xx_i$
the $i$-th row of the shift $\xx$. Let $\pi_i$ be a permutation of $[m]$
such that $\vc s_i^{\pi_i(1)}\geq\cdots\geq\vc s_i^{\pi_i(m)}$. Repeating the above procedure
with this $1$-dimensional data we see that $f_i(r_1,\dots,r_m)$ is concave. So
$f(r_1,\dots,r_m)$ is also concave in the $r_k$, being the following sum of concave functions,
$$f(r_1,\dots,r_m)\ =\ \vc c\xx\ =\ \sum_{i=1}^d\vc c_i\xx_i\ =\
   \sum_{i=1}^df_i(r_1,\dots,r_m)\ .$$
This also shows that if $-\vc c$ is shifted then $-f$ is concave and hence $f$ is convex.

\smallskip
We proceed with the (positive) algorithmic statements of the theorem. For part 1,
which was also proved in \cite{LevinO:2016}, just note that for fixed $m$, there are
$\Oh(r^{m-1})$ feasible solutions in \eqref{IP}, obtained by taking integers
$0\leq r_1,\dots,r_{m-1}\leq r$ with $\sum_{i=1}^{m-1} r_i\leq r$ and setting
$r_m:=r-\sum_{i=1}^{m-1} r_i$.
Hence, in polynomial time we can enumerate all, pick the best,
and obtain an optimal solution $\vc x(r_1,\dots,r_m)$ to the shifted problem \eqref{shiftZ}.

For part 2, if $\vc c$ is shifted, then, as just shown, $f$ is concave.
So the integer program \eqref{IP} is to maximize a concave function with the number $m$ of
variables as a parameter. By known results on convex integer minimization, see \cite{OertelWW:14},
this problem is fixed-parameter tractable and solvable 
in time $p(m)(\log r)^q$ for some computable function $p$ of $m$ 
and some constant~$q$. 
Because it is enough to present the objective function by an oracle, we can
extend this to the case when $\vc c$ is not given explicitly, but by a
partial sums oracle $\gamma(i,j) = \sum_{\ell=1}^j c_i^\ell$, in which case
$r$ can be given in binary.  

For part 3, if $-\vc c$ is shifted, then, as just shown, $f$ is convex. Therefore,
the maximum in \eqref{IP} is attained at a vertex of the simplex. These vertices are
the $m$ vectors $r\vc e_1,\dots,r\vc e_m$, where $\vc e_k$ is the $k$-th unit vector in $\R^m$.
Hence, in polynomial time we can pick the best vertex $r\vc e_k$
and obtain again an optimal solution $\vc x(r\vc e_k)$ to \eqref{shiftZ}.

\smallskip
To complete the proof of the theorem, we return to the hardness claim in part~1.
We proceed
by reduction from the \textsc{Multidemand Set Cover} (MSC) problem, which is
as follows.  Given is a universe $U = \{u_1, \dots, u_n\}$,
a collection of multidemands $\vc d=(d_1, \dots, d_n)$ where
$d_i\subseteq \N$ for $i=1, \dots, n$, a covering set system 
$\mathcal{F} = \{F_1, \dots, F_k\} \subseteq 2^U$, and an integer $r \in \N$.
The goal is to find an integer partition $r = p_1 + \dots + p_k$ such that, 
for all $i=1, \dots, n$, we have $(\sum_{j:u_i \in F_j} p_j) \in d_i$.  
Knop et al.~\cite{KnopKMT:16} prove that MSC is \W{1}-hard with respect to the
parameter~$n$ even when~$n=k$.

For clarity, we briefly and informally remark on a meaning of the MSC problem.
We wish to take each set $F_j$, $j\in[k]$, with multiplicity $p_j$,
and we demand that for each universe element $u_i$, $i\in[n]$,
the total sum of multiplicities of sets $u_i$ belongs to, falls into the
constraint set~$d_i$.

Given an instance $U, \vc d, \mathcal{F}$ and $r$ of MSC, let $S \subseteq
\{0,1\}^n$ be the set of characteristic vectors of $\mathcal{F}$,
where $|S|=m$ in our case.  
We will define $\vc c$ inductively.  Fix a row $i \in [n]$ and let $c_i^1 :=1$
if $1\in d_i$ and $c_i^1 :=0$ otherwise.
For $j=1, \dots, r-1$, let $c_i^{j+1} := 1-\sum_{\ell=1}^j c_i^\ell$ if $j+1 \in d_i$,
and $c_i^{j+1} := -\sum_{\ell=1}^j c_i^\ell$ otherwise.
Then $\max \vc c \vc \xx \leq n$ and $\vc c \vc \xx= n$
exactly when the multiplicities $p_1, \dots, p_m$ of the vectors of
$S$ in $\vc x\in S^r$ are such that the number of $1$'s in each row $i$ of
$\vc \xx$ falls into $d_i$, by our choice of $\vc c$.
This is the case if and only if the MSC instance is a ``yes'' instance.
\epr

\smallskip
In the rest of this section we provide several supplementary results
related to the cases of Theorem~\ref{main}.

\end{accumulate}

Let us first give an exemplary application of part 2 of Theorem~\ref{main} now. 
Bredereck et al.~\cite{BredereckFNST:2015} study the
\textsc{Weighted Set Multicover} (WSM) problem, which is as follows. 
Given a universe $U=\{u_1, \dots, u_k\}$, integer demands $d_1, \dots, d_k
\in \N$ and a multiset $\FF = \{F_1, \dots, F_n\} \subseteq 2^U$ with
weights $w_1, \dots, w_n \in \N$, find a multiset $\FF' \subseteq \FF$ of
smallest weight which satisfies all the demands -- that is, for all $i=1, \dots,
k$, $|\{F \in \FF' \mid u_i \in F\}| \geq d_i$.
It is shown~\cite{BredereckFNST:2015} that this problem is \FPT when the size of the universe is a
parameter, and then several applications in computational social choice are
given there.

Notice that $\FF$ can be represented in a succinct way by viewing $\FF$ as a
set $\FF_s = \{F_1, \dots, F_K\}$ and representing the different copies of
$F \in \FF_s$ in $\FF$ by defining $K$ weight functions $w_1, \dots, w_K$
such that, for each $i=1, \dots, K$, $w_i(j)$ returns the total weight of
the first $j$ lightest copies of $F_i$, or $\infty$ if there are less than
$j$ copies.  We call this the \textit{succinct variant}.

Bredereck et al.~\cite{BredereckFNST:2015} use Lenstra's algorithm for their result, which only works when $\FF$ is given explicitly.
We note in passing that our approach allows us 
to extend their result to the succinct case.

\begin{proposition}\label{prop:wsm_fpt}\APXmark
Weighted Set Multicover is in \FPT with respect to universe size~$k$, even in the succinct variant.
\end{proposition}

\begin{accumulate}
\ifthenelse{\value{Accumulate} = 1}{%
  \proofof{Proposition~\ref{prop:wsm_fpt}}
}{\bpr}
Let $\vc d, \FF, \vc w$ be an instance of \textsc{Weighted Set Multicover} with universe of size $k$, where $k$ is parameter, and let $K = |\FF_s| \leq 2^k$.
We will construct an SCO instance with $S$ of size $k+K$ such that solving it will correspond to solving the original WSM problem. Since $\max c \xx$ with $c=\cc$ is equivalent to $\min c \xx$ with $c = -\cc$, we will define a minimization instance with $c$ non-decreasing.

Let $S = \big(\{(\vc f_i, \vc e_i) \mid F_i \in \FF_s\} \cup \{(\vc 0, \vc
0)\} \big) \subseteq \{0,1\}^{k+K}$ where $f_i$ is the characteristic vector
of the set $F_i$ and $e_i$ is the $i$-th unit vector.  Let $W$ be the total
weight of $\FF$ and let $D = \sum d_i$ be the total demand.  Then, the first
$k$ rows of $\vc c$ are defined as $c_i^j = -W$ if $j \leq d_i$ and $c_i^j =
0$ otherwise. The corresponding partial sums oracle is $\gamma(i,j) =
-jW$ if $j \leq d_i$ and $\gamma(i,j) = -d_iW$ otherwise.  The remaining $K$
rows of $\vc c$ are exactly the weight functions $w_1, \dots, w_K$, that is,
$\gamma(k+i,j) = w_i(j)$, where $w_i$ returns $DW$ whenever it should return
$\infty$.

Let $r=D$ and solve the SCO given above. A solution is represented by multiplicities $r_i$ for $i=0, \dots, K$, where $r_0$ is the multiplicity of the $\vc 0$ vector.
We interpret it as a WSM solution straightforwardly: $r_i$ means how many copies of $F_i$ we take to the solution, and we always choose the $r_i$ lightest.
Observe that if the objective is at most $-(D-1)W$, it means that the
solution ``hits'' all the $-W$ items in the first $k$ rows, which in turn
means all demands are satisfied.
On the other hand, the objective is more
than $-(D-1)W$ only if some demand was not satisfied.  Also, if $\min c \xx
\leq -(D-1)W$, then the solution never hit a $DW$ item in the last $K$ rows,
which in turn means that, for each $i=1, \dots, K$, we have never used more
copies of $F_i$ than there actually are.  Because the objective decomposes
into $-DW + \sum_{i=1}^K w_i(r_i)$, where the first term is a constant and
the second is exactly the weight of the solution, we have found the optimum.
\epr
\end{accumulate}

Theorem~\ref{main}, part 3, can be applied also to sets $S$ presented implicitly by
an oracle. 
A {\em linear optimization oracle} for $S\subseteq \Z^n$ is one that, queried
on $w \in \Z^n$, solves the linear optimization problem 
$\max \{\vc w\vc s \mid \vc s\in S\}$. Namely, the oracle either
asserts that the problem is infeasible, or unbounded, or provides an optimal solution.
As mentioned before, even for $r=2$, the shifted problem for perfect matchings is
\NP-hard, and hence for general $\vc c$ the shifted problem over $S$ presented by a linear
optimization oracle is also hard even for $r=2$. 
In contrast, we have the following strengthening.

\bt{shiftedt}\APXmark
The shifted problem \eqref{shiftZ} with $\vc c$ nondecreasing, over any set
$S\subset\Z^n$ which is presented by a linear optimization oracle, 
can be solved in polynomial time.
\et
\begin{accumulate}
\ifthenelse{\value{Accumulate} = 1}{%
  \proofof{Theorem~\ref{shiftedt}}
}{\bpr}
Let $\vc w:=\sum_{j=1}^r\vc c^j$ be the sum of the columns of $\vc c$, and query the linear optimization
oracle of $S$ on $\vc w$. If the oracle asserts that the problem is
infeasible, then $S=\emptyset$ hence $S^r=\emptyset$ hence so is the shifted problem. 
Suppose it asserts that the problem is unbounded.
Then for every real number $q$ there is an $\vc s\in S$ with $\vc w
\vc s \geq q$. Then the matrix
$\vc x:=[\vc s,\dots,\vc s]$ with all columns equal to $\vc s$ satisfies $\vc \xx=\vc x$ and hence 
$\vc c\xx=\vc c \vc x=\sum_{j=1}^r\vc c^js=\vc w\vc s\geq q$, and therefore the shifted problem is also unbounded.

Suppose then that the oracle returns an optimal solution $\vc s^*\in S$ and define
$\vc x^*:=[\vc s^*,\dots,\vc s^*]$ to be the matrix with all columns equal to $\vc s^*$. We claim that $\vc x^*$
is an optimal solution to the shifted problem. Suppose indirectly $\vc x$ is a
strictly better solution.
Let $T$ be the set of columns of $\vc x$, that means 
$T:=\{\vc x^1,\dots,\vc x^r\}=\{\vc t^1,\dots,\vc t^m\}$ 
for suitable distinct $\vc t^k\in S$, where $k=1,\dots,m$ and $m=|T|\leq r$. 

Consider the shifted problem over $T$. By the proof of the algorithmic part 2 of
Theorem \ref{main}, we will have an optimal solution 
$\vc y:=\vc t(r\vc e_k)=\vc t(0,\dots,0,r,0\dots,0)$
for some unit vector $\vc e_k\in\R^m$, that is, $\vc y=[\vc t,\dots,\vc t]$ for some $\vc t\in T$.
We then obtain 
$$\vc w\vc t\ =\ \vc c\vc y\ =\ \vc c\vc {\overline
  y}\ \geq\ \vc c\xx\ >\ \vc c\xx^*\ =\ \vc c\vc x^*\ =\ \vc w\vc s^*$$
which is a contradiction to the assumed optimality of $\vc s^*$, completing the proof.
\epr
\end{accumulate}

\section{MSO-definable Sets: XP for Bounded Treewidth}
\label{sec:xp_sco}

In this section we study another tractable and rich case of shifted
combinatorial optimization, namely that of the set $S$ defined in the MSO
logic of graphs.
This case, in particular, includes well studied MSO-partitioning framework
of graphs (see below) which is tractable on graphs of bounded treewidth
and clique-width.
In the course of proving our results, it is useful to study a geometric connection of 
$0/1$ SCO problems to separable optimization over decomposable polyhedra.

\subsection{Relating SCO to decomposable polyhedra}\label{sec:geometry}
\begin{accumulate}
\ifthenelse{\value{Accumulate} = 1}{
 \section{Additions to Subsection~\ref{sec:geometry}}}{}
\end{accumulate}

The purpose of this subsection is to demonstrate how shifted optimization
over $0/1$ polytopes closely relates to an established concept of decomposable
polyhedra.
We refer to\ifthenelse{\value{Accumulate} = 0}{ the book of}{}
Ziegler~\cite{Ziegler:95} for definitions and terminology
regarding polytopes.

\begin{definition}[Decomposable polyhedron and Decomposition oracle]
  \label{def:decomposable}
  A polyhedron $P \subseteq \R^n$ is \emph{decomposable} if for
  every $k \in \N$ and every
  $\vc x \in kP \cap \Z^n$, there are $\vc x^1, \ldots,\vc x^k \in P \cap \Z^n$
  with $\vc x = \vc x^1 + \cdots +\vc x^k$, 
  where $kP = \{k\vc y \mid\vc y \in P\}$.
\ifthenelse{\value{Accumulate} = 1}{}{\par}
  A \emph{decomposition oracle} for a decomposable $P$ is one that,
  queried on $k \in \N$ given in unary
  and on $\vc x \in kP \cap \Z^n$, returns $\vc x^1, \ldots,\vc x^k \in
  P \cap \Z^n$ with $\vc x =\vc x^1 + \cdots +\vc x^k$.
  \end{definition}
This property is also called \emph{integer decomposition property} or being
\emph{integrally closed} in the literature.  
The best known example are polyhedra given by totally unimodular
matrices~\cite{BT:78}. 
Furthermore, we will use the following notion.

\begin{definition}[Integer separable (convex) minimization oracle]
  \label{def:separable}
  Let $P \subseteq \R^n$ and let $f(\vc x) = \sum_{i=1}^n f_i(x_i)$ be a
  separable function on $\R^n$. An
  \emph{integer separable minimization oracle} for $P$ is one that,
  queried on this $f$, either reports that $P \cap \Z^n$ is
  empty, or that it is unbounded, or returns a point $\vc x \in P \cap \Z^n$
  which minimizes $f(\vc x)$.
\ifthenelse{\value{Accumulate} = 1}{}{\par}
  An \emph{integer separable convex minimization oracle} for $P$ is
  an integer separable minimization oracle for $P$ which can only be
  queried on functions $f$ as above with all $f_i$ convex.
\end{definition}

We now formulate how these notions naturally connect with 
\ifthenelse{\value{Accumulate} = 1}{SCO}
{shifted optimization in the case when $S$~is~$0/1$}.

\begin{lemma}\APXmark\label{lem:idp_separable}
  Let $(S,\vc c, r)$ be an instance of shifted combinatorial
  optimization, with $S \subseteq \{0,1\}^n$, $r\in\N$ and $\vc c\in\Z^{n\times r}$.
  Let $P \subseteq [0,1]^n$ be a
  polytope such that $S = P \cap \{0,1\}^n$ and let $Q \subseteq [0,1]^{n+n'}$
  be some extension of $P$, that is, $P = \{\vc x \mid (\vc x,\vc y) \in
  Q\}$.

  Then, provided a decomposition oracle for $Q$ and an integer
  separable minimization oracle for $rQ$, the shifted problem given by
  $(S,\vc c, r)$ can be solved with one call to the optimization oracle and
  one call to the decomposition oracle. Furthermore, if $\vc c$ is shifted,
  an integer separable convex minimization oracle suffices.
\end{lemma}

To demonstrate Lemma~\ref{lem:idp_separable} we use it to give an alternative proof of
the result of Kaibel et al.~\cite{KaibelOS:2015} that the shifted
problem is polynomial when $S = \{\vc x \mid \vc A\vc x =\vc b,
\, \vc x \in \{0,1\}^n\}$ and $\vc A$ is totally unimodular. 
It is known that $P = \{\vc x \mid \vc A\vc x= \vc b,
\mathbf{0} \leq\vc x \leq \mathbf{1}\}$ is decomposable and a decomposition
oracle is realizable in polynomial
time~\cite{Schrijver:2003}. Moreover, it is known that an integer
separable convex minimization oracle for $rP$ is realizable in
polynomial time~\cite{HochbaumS:90}. Lemma~\ref{lem:idp_separable}
implies that the shifted problem is polynomial for this $S$ when $\vc c$
is shifted.

The reason we have formulated Lemma~\ref{lem:idp_separable} for $S$ given
by an extension $Q$ of the polytope $P$ corresponding to $S$, is the
following:  while $P$ itself might not be decomposable, 
there always exists an extension of it which is decomposable.
\ifthenelse{\value{Accumulate} = 1}{See the Appendix for more details.}{}
\begin{accumulate}
\ifthenelse{\value{Accumulate} = 1}{We start with proving the promised
fact of existence of a decomposable extension.}{}
  \begin{proposition}
  Let $P \subseteq \R^n$ be a $0/1$ polytope with $m$ vertices. Then it has a decomposable
  extension $Q \subseteq \R^{n+m}$.
  \end{proposition}
  \begin{proof}
  Let $V = \{\vc v_1, \ldots, \vc v_m\} \subseteq \{0,1\}^n$ be the vertices
  of $P$. Then $P$ can be obtained by projecting the $m$-dimensional
  simplex:
  \begin{align*}
    \lambda_1 + \cdots + \lambda_m &= 1 \\
    \lambda_1 \vc v_1 + \cdots + \lambda_m \vc v_m &= \vc x \\
    \lambda_i & \geq 0 \qquad  \forall i \in \{1, \ldots, m\}
  \end{align*}
  Let $Q$ be the polytope defined by the system above and fix $k \in \N$. Consider an
  integer point $(\vc x, \vc \lambda) \in kQ$. Then $\sum_{i=1}^m \lambda_i (\vc v_i, \vc e_i)$ (with $\vc
  e_i$ the $i$-th unit vector) is a decomposition of $(\vc x, \vc
  \lambda)$ into vertices of $Q$, certifying that $Q$ is decomposable.
  \qed\end{proof}
\end{accumulate}

Other potential candidates where Lemma~\ref{lem:idp_separable} could
be applied are classes of polytopes that are either decomposable, or
allow efficient integer separable (convex) minimization. Some known
decomposable polyhedra are stable set polytopes of perfect graphs, polyhedra
defined by $k$-balanced matrices~\cite{Zambelli:2007}, polyhedra defined by
nearly totally unimodular matrices~\cite{Gijswijt:2005}, etc. Some known cases
where integer separable convex minimization is polynomial are for
$P=\{\vc x \mid \vc A\vc x=\vc b, \vc x \in \{0,1, \dots, r\}^n\}$ 
where the Graver basis of $\vc A$
has small size or when $\vc A$ is highly structured, namely when $\vc A$ is
either an $n$-fold product, a transpose of it, or a 4-block $n$-fold
product; see the books of Onn~\cite{Onn:2010} and De Loera, Hemmecke and
Köppe~\cite{DeLoeraHK:2013}.

\begin{accumulate}
Now we return to the proof of Lemma~\ref{lem:idp_separable}.
  
\begin{proofof}{Lemma~\ref{lem:idp_separable}}
  Fix $(\vc x,\vc y) \in rQ \cap \Z^{n+n'}$ and consider the set
\begin{align*}
  Q_{(\vc x,\vc y)}^r&:= \left\{ (\vc a,\vc b) = ((\vc a^i, \vc b^i))_{i=1}^r \left|~
  \sum_{i=1}^r (\vc a^i, \vc b^i) = (\vc x,\vc y), (\vc a^i, \vc b^i) \in Q \cap \{0,1\}^{n+n'},
  1 \leq i \leq r \right\}\right. \\
  &\, \subseteq \left(Q \cap \{0,1\}^{n+n'}\right)^r
\end{align*}

Note that $\vc c \vc{\la a}$ is the same for all
$(\vc a,\vc b) \in Q_{(\vc x,\vc y)}^r$, and observe that $(Q \cap \{0,1\}^{n+n'})^r =
\bigcup_{(\vc x,\vc y) \in (rQ \cap \{0,1\}^{n+n'})} Q_{(\vc x,\vc y)}^r$.

Consider now the objective function $\vc c \in \Z^{n \times r}$. 
Define $w_i: \{0, \dots, r\} \rightarrow \Z$ for $i=1, \dots,
n$ by $w_i(k) := \sum_{j=1}^k c_i^j$. Note that if $\vc c$ is shifted,
every $w_i$ is concave as it is a partial sum of a non-increasing
sequence. Observe that $\vc c \vc{\la a} =
\sum_{i=1}^n w_i(\sum_{j=1}^r a_i^j) = \sum_{i=1}^n w_i(x_i)$. It
follows that the minimum of $\sum_{i=1}^n w_i(x_i)$ over $rQ \cap \Z^{n+n'}$ equals
the minimum of $\vc c \vc{\la a}$ over $(Q \cap \{0,1\}^{n+n'})^r$. 

To solve the shifted problem, let $f=\sum_{i=1}^{n+n'} f_i$ with $f_i
:= w_i$ for $1 \leq i \leq n$ and $f_i := 0$ for $n+1 \leq i \leq
n+n'$, and query the integer separable minimization oracle on $rQ$ with
$-f$ (minimizing $-f$ maximizes $f$). The oracle returns that the
problem is either infeasible or unbounded or returns $(\vc x,\vc y)
\in rQ$ maximizing $f$. Next, query the decomposition oracle for $Q$
on $r$ and $(\vc x,\vc y)$ to obtain $((\vc x^i,\vc y^i))_{i=1}^{r}$,
and return this as the solution. If $\vc c$ is shifted then $-f$ is
convex and an integer separable convex minimization oracle suffices
for the first step.
\end{proofof}

\end{accumulate}

\subsection{XP algorithm for MSO-definable set}
\label{sec:XPMSO}
\begin{accumulate}
 \ifthenelse{\value{Accumulate} = 1}{
\section{Additions to Subsection~\ref{sec:XPMSO}}
}{}\end{accumulate}

We start with defining the necessary specialized terms.
\ifthenelse{\value{Accumulate} = 1}{%
We assume that the reader is familiar with the standard 
{\em treewidth} of a graph (see the Appendix).
}{}
\begin{accumulate}
\begin{definition}[Treewidth]
Given a graph~$G$, a \emph{tree-decomposition of $G$} is an ordered pair
$(T,\mathcal{W})$, where~$T$ is a tree and 
$\mathcal{W} = \{W_x \subseteq V(G) \mid x \in V(T)\}$ is a collection
of {\em bags} (vertex sets of~$G$), 
such that the following hold:
\begin{enumerate}
    \item $\bigcup_{x \in V(T)} W_x = V(G)$;
    \item for every edge~$e = uv$ in~$G$, there exists~$x \in V(T)$ such that~$u,v \in W_x$;
    \item for each~$u \in V(G)$, the set~$\{x \in V(T) \mid u \in W_x\}$
    induces a subtree of $T$.
\end{enumerate}
The {width} of a tree-decomposition $(T,\mathcal{W})$ is 
$(\max_{x \in V(T)}|W_x|)-1$.
The \emph{treewidth} of~$G$, denoted~$tw(G)$, is the smallest width of a
tree-decomposition of~$G$. 
\end{definition}
\end{accumulate}

Given a matrix $\vc A \in \Z^{n \times m}$, we define the corresponding
\emph{Gaifman graph} $G = G(\vc A)$ as follows. Let $V(G) = [m]$.
We let $\{i,j\} \in E(G)$ if and only if there is an
$r \in [n]$ with $\vc A[r, i] \neq 0$ and $\vc A[r, j] \neq 0$.  
Intuitively, two vertices of $G$ are adjacent if the corresponding 
variables $x_i,x_j$ occur together in some constraint
of $\vc A\vc x \leq \vc b$.
The \emph{(Gaifman) treewidth of a matrix $\vc A$} is then the
treewidth of its Gaifman graph, i.e., $tw(\vc A) := tw(G(\vc A))$.

The aforementioned MSO-partitioning framework
of graphs comes as follows.

\vskip.2cm\noindent{\bf MSO-partitioning problem.}
  Given a graph $G$, an MSO$_2$ formula $\varphi$ with one free vertex-set
  variable and an integer $r$, the task is as follows;
\begin{itemize}
\item to find a
  partition $U_1 \dot\cup\, U_2 \dots \dot\cup\, U_r = V(G)$ of the vertices of $G$
  such that $G \models \varphi(U_i)$ for all $i=1,\dots,r$, or
\item to confirm that no such partition of $V(G)$ exists.
\end{itemize}

For example, if $\varphi(X)$ expresses that $X$ is an independent set,
then the $\varphi$-\textsc{MSO-partitioning} problem decides if
$G$ has an $r$-coloring, and thus, finding minimum feasible $r$ (simply 
by trying $r=1,2,\dots$) solves the \textsc{Chromatic number} problem. 
Similarly, if $G \models \varphi(X)$ when $X$
is a dominating set, minimizing $r$ solves the 
\textsc{Domatic number} problem, and so on.

Rao~\cite{Rao:2007} showed an algorithm for \textsc{MSO-partitioning},
for any MSO$_2$ formula $\varphi$,
on a graph $G$ with treewidth $tw(G) = \tau$ running in time
 $r^{f(\varphi, \tau)}n$ (\XP) for some computable function $f$. 
Our next result widely generalizes this to SCO over MSO-definable sets.

\begin{definition}[MSO-definable sets]
For a graph $G$ on $|V(G)| = n$ vertices,
we interpret a $0/1$ vector $\vc{x} \in \{0,1\}^n$ as the set $X \subseteq V$
where $v \in X$ iff $x_v = 1$.
We then say that \emph{$\vc{x}$ satisfies a formula $\varphi$} if $G \models
 \varphi(X)$. Let 
$$S_{\varphi}(G) = \{\vc{x} \mid \vc{x} \textrm{ satisfies } \varphi
	\textrm{ in $G$}\} \,.$$
\end{definition}

Let $\vc c$ be defined as $c_i^1 := 1$ for $1 \leq i \leq
n$ and $c_i^j := -1$ for $2 \leq j \leq r$ and $1 \leq i \leq n$. 
Observe then the following: deciding whether the shifted problem with $S =
S_{\varphi}(G)$, $\vc c$ and $r$, has an optimum of value $n$ is
equivalent to solving the \textsc{MSO-partitioning} problem for $\varphi$.

\begin{theorem}\APXmark\label{thm:sco_xp}
Let $G$ be a graph of treewidth $tw(G) = \tau$,
let $\varphi$ be an MSO$_2$ formula and
$S_\varphi(G) = \{\mathbf{x} \mid \mathbf{x} \textrm{ satisfies }
\varphi \}$. There is an algorithm solving the shifted problem with $S
= S_\varphi(G)$ and any given $\vc c$ and $r$ in time $r^{f(\varphi,
  \tau)}\cdot|V(G)|$ for some computable function $f$. In other words, for
parameters $\varphi$ and $\tau$, the problem is in the complexity class \XP.
\end{theorem}

We will prove Theorem~\ref{thm:sco_xp} using Lemma~\ref{lem:idp_separable}
on separable optimization over decomposable polyhedra. 
To that end we need to show the following two steps:
\begin{enumerate}\vspace*{-1ex}
\item
There is an extension $Q$ of the polytope $P = \conv(S_\varphi(G))$ which is
decomposable and endowed with a decomposition oracle (Definition~\ref{def:decomposable}),
 and
\item
there is an integer separable minimization oracle
(Definition~\ref{def:separable}) for the polytope~$rQ$.
\end{enumerate}
The first point is implied by a recent
result of Kolman, Koutecký and Tiwary~\cite{KolmanKT:2015}: 

\begin{proposition}[\cite{KolmanKT:2015}]\label{prop:courcelle_lp}
  Let $G$ be a graph on $n$ vertices of treewidth $tw(G) = \tau$, 
  and $\varphi$ be an $MSO_2$ formula with one free vertex-set variable.
  Then, for some computable functions $f_1,f_2,f_3$,
  there are matrices $\vc A,\vc B$ and a vector $\vc{b}$,
  computable in time $f_1(\varphi, \tau)\cdot n$, such that
      \begin{enumerate}\vspace*{-1ex}
  \item the polytope $Q = \{(\bfx, \mathbf{y}) \mid \vc A\mathbf{x} +
    \vc B\mathbf{y} = \mathbf{b}, \mathbf{y} \geq \mathbf{0}\} \subseteq
    [0,1]^{f_2(\varphi,\tau)n}$ is an extension of the polytope
    $P = \conv(S_\varphi(G))$,
  \item $Q$ is decomposable and endowed with a decomposition oracle, and
  \item the (Gaifman) treewidth of the matrix $(\vc A\,\vc B)$
    is at most $f_3(\varphi, \tau)$.
  \end{enumerate}
\end{proposition}

The second requirement of Lemma~\ref{lem:idp_separable} follows from
efficient solvability of the constraint satisfaction problem (CSP) of bounded
treewidth, originally proven by Freuder~\cite{Freuder:90}. We will use a
natural weighted version of this folklore result.  

\begin{accumulate}
\begin{definition}[CSP]
An instance $I=(V,\DD,\HH,\CC)$ of CSP consists of
\begin{itemize}
\item a set of {\em variables} $z_v$, one for each $v\in V$; without loss of 
generality we assume that $V = \{1,\ldots,n\}$,
\item a set $\DD$ of finite {\em domains} $D_v\subseteq \Z$ (also 
denoted $D(v)$), one for each $v\in V$,
\item a set of {\em hard constraints} $\HH \subseteq \{C_{U} \mid U \subseteq V \}$ where each
hard constraint $C_{U} \in \HH$ with $U=\{i_1, i_2,\dots,i_k\}$ and 
$i_1 < \cdots < i_k$, is a $|U|$-ary relation
$C_U \subseteq D_{i_1}\times D_{i_2}\times \cdots \times D_{i_k}$,
\item a set of {\em weighted soft constraints} $\CC \subseteq \{w_U \mid
  U \subseteq V\}$ where each $w_U \in \CC$ with $U=\{i_1, i_2,\dots,i_k\}$ and 
$i_1 < \cdots < i_k$ is a function $w_U:D_{i_1}\times D_{i_2}\times \ldots\times D_{i_k}\rightarrow \R$.
\end{itemize}
For a vector $z^{}=(z^{}_1, z_2, \ldots,z^{}_n)$ and a set $U=\{i_1,
i_2,\dots,i_k\}\subseteq V$ with $i_1 < i_2 < \cdots < i_k$, we define
the {\em projection of} $z$ on $U$ as 
$z^{}|_U=(z^{}_{i_1}, z_{i_2}, \ldots, z^{}_{i_k})$.
A vector $z\in \Z^n$ {\em satisfies the hard constraint} $C_U \in \HH$ if and only if $z|_U \in C_U$. 
We say that a vector
$z^{\star}=(z^{\star}_1,\ldots,z^{\star}_n)$ is {\em a feasible assignment} for
$I$ if $z^{\star} \in D_1\times D_2\times \ldots\times D_n$ and 
$z^{\star}$ satisfies every hard constraint $C\in \HH$. The
\emph{weight} of $z^{\star}$ is $w(z^{\star})
= \sum_{w_U \in \CC}w_U(z^{\star}|_U)$. 
\end{definition}
\end{accumulate}

For a CSP instance $I=(V,\DD,\HH,\CC)$ one can define
the \emph{constraint graph} of $I$ as $G=(V,E)$
where $E= \{\{u,v\} \mid (\exists C_{U} \in \HH) \vee (\exists w_U \in
\CC)\textrm{ s.t. } \{u,v\} \subseteq U\}$. 
The {\em treewidth of a CSP instance $I$} is defined as
the treewidth of the constraint graph of $I$.

\begin{proposition}[\cite{Freuder:90}]\label{prop:csp_tw}
  Given a CSP instance $I$ of treewidth $\tau$ and maximum domain size
  $D=\max_{u\in V}|D_u|$, 
  a minimum weight solution can be found in time $\Oh(D^\tau (n +
  |\HH| + |\CC|))$.
  \end{proposition}

Proposition~\ref{prop:csp_tw} can be used to realize an integer separable
minimization oracle for integer programs of bounded treewidth, as follows.

\begin{lemma}\APXmark\label{lem:ilp_tw}
  Let $\vc A \in \Z^{n \times m}, \mathbf{b} \in \Z^m,
  \vcl,\mathbf{u} \in \Z^n$ be given s.t.\ $tw(\vc A) = \tau$,
  and let $D = \|\mathbf{u} - \vcl\|_{\infty}$. Then an integer
  separable minimization oracle over $P = \{\bfx \mid \vc A\bfx=\mathbf{b},
  \vcl \leq \vc x \leq \mathbf{u}\}$ is realizable in time $D^\tau
  (n+m)$.
\end{lemma}
\ifthenelse{\value{Accumulate} = 1}{\noindent
  The proof of Lemma~\ref{lem:ilp_tw} proceeds by constructing a CSP instance $I$ based on the
  ILP~$\vc A\bfx = \mathbf{b}$, $\vcl \leq \bfx \leq \mathbf{u}$, such
  that solving $I$ corresponds to integer separable minimization over~$P$.
  Since the treewidth of $I$ is $\tau$ and the maximum domain size is
  $D$, Proposition~\ref{prop:csp_tw} does the job.
}{}
\begin{accumulate}

\ifthenelse{\value{Accumulate} = 1}{%
  \proofof{Lemma~\ref{lem:ilp_tw}}
}{\bpr}
  The proof proceeds by constructing a CSP instance $I$ based on the
  ILP $\vc A\bfx = \mathbf{b}, \vcl \leq \bfx \leq \mathbf{u}$, such
  that solving $I$ corresponds to integer separable minimization over
  $P$. Since the treewidth of $I$ is $\tau$ and the maximum domain size is
  $D$, Proposition~\ref{prop:csp_tw} does the job.

  First, let $V = \{x_1, \dots, x_n\}$. Then, for every $i=1, \dots,
  n$, let $D_i = \{\ell_i, \ell_i+1, \dots, u_i\}$ and $\DD = \{D_i \mid i=1,
  \dots, n\}$. Observe that $\max_i |D_i| = \|\mathbf{u} - \vcl\|_{\infty} =
  D$. Regarding hard constraints $\HH$, observe that every row $\vc a_j$
  of $\vc A$ contains at most $\tau+1$ non-zeros, since otherwise the
  Gaifman graph of $\vc A$ would contain a clique of size $\tau+2$,
  contradicting its treewidth of $\tau$. Let $U_j = \{i_1, \dots,
  i_k\}$, where $k \leq \tau+1$, be the set of indices of non-zero
  elements of $\vc a_j$, and let $x_c = 0$ for all $c \not\in U_j$. Let
  $C_{U_j}$ be the set of assignments from $D_{i_1} \times \cdots
  \times D_{i_k}$ to $x_{i_1}, \ldots, x_{i_k}$ that satisfy 
  $\vc a_j \bfx = b_j$; obviously $|C_{U_j}|
  \leq D^k$ and it can be constructed in time $\Oh(D^k)$. Then, $\HH =
  \{C_{U_j} \mid j = 1, \dots, m\}$. Finally, for a given separable function
  $f$ such that $f(\vc x) = \sum_{i=1}^n f_i(x_i)$, let $\CC =
  \{w_{\{x_i\}} \mid
  i=1, \dots, n\}$ where $w_{\{x_i\}} = f_i$ for all $i$.

  It is easy to verify that the feasible assignments of $I$ correspond
  to integer solutions of $\vc A\bfx = \mathbf{b}, \vcl \leq \bfx
  \leq \mathbf{u}$, that its maximum domain size is $D$ and its weight
  function $w$ is exactly~$f$. Finally, the treewidth of $I$ is
  $\tau$, since the Gaifman graph of $G(I)$ of $I$ is exactly
  $G(\vc A)$. Then Proposition~\ref{prop:csp_tw} solves $I$ in time
  $\Oh(D^\tau (n + |\HH| + |\CC|) = \Oh(D^\tau (n+m))$, concluding the
  proof.
\epr
\end{accumulate}

Consequently, we can finish the proof of Theorem~\ref{thm:sco_xp} by using
Lemma~\ref{lem:idp_separable}.

\begin{accumulate}

\begin{proofof}{Theorem~\ref{thm:sco_xp}}
  By Proposition~\ref{prop:courcelle_lp} there is a computable function $f$
  such that there exists a polytope $Q = \{(\bfx, \mathbf{y}) \mid
  \vc A\mathbf{x} + B\mathbf{y} = \mathbf{b}, \mathbf{y} \geq \mathbf{0}\}
  \subseteq [0,1]^{f_2(\varphi, \tau)n}$ which is an extension of
  $P_\varphi(G)$; let $F:=f_3(\varphi, \tau)$. Also, $S_\varphi(G) =
  P_\varphi(G) \cap \{0,1\}^n$. Moreover, $Q$ is decomposable, a
  decomposition oracle for $Q$ is realizable in polynomial time, and
  the treewidth of the matrix $(\vc A \,\vc B)$ is at most $F$. Note that $rQ$
  is given by $\vc A\vc x + B\vc y = r\vc b, \mathbf{0} \leq (\bfx,\mathbf{y}) \leq
  (r, \dots, r)$, so its treewidth is $F$ as
  well. Lemma~\ref{lem:ilp_tw} realizes an integer separable
  minimization oracle for $rQ$ in time $r^F(n+Fn)$. Since all
  conditions of Lemma~\ref{lem:idp_separable} are met, this concludes
  the proof.
  \end{proofof}
\end{accumulate}

\smallskip
Besides treewidth, another useful width measure of graphs is the
{\em clique-width} of a graph~$G$.
Rao's result~\cite{Rao:2007} applies also to the MSO-partitioning problem
for MSO$_1$ formulas and graphs of bounded clique-width.
We show the analogous extension of Theorem~\ref{thm:sco_xp} next.

\begin{definition}[Clique-width]\label{def:cw}
This is defined for a graph $G$ as the smallest number of labels $k=cw(G)$
such that some labeling of $G$ can be constructed
by an algebraic {\em k-expression} using the following operations
(where $1\leq i,j\leq k$):
\begin{enumerate}\vspace*{-1ex}
\item create a new vertex with label $i$;
\item take the disjoint union of two labeled graphs;
\item add all edges between vertices of label $i$ and label $j$; and
\item relabel all vertices with label $i$ to have label $j$.
\end{enumerate}
\end{definition}

\begin{corollary}\APXmark\label{cor:sco_xp_cw}
Let $G$ be a graph of clique-width $cw(G) = \gamma$
given along with a $\gamma$-expression (cf.~Definition~\ref{def:cw}),
let $\psi$ be an MSO$_1$ formula and
$S_\psi(G) = \{\mathbf{x} \mid \mathbf{x} \textrm{ satisfies }\psi \}$. 
There is an algorithm solving the shifted problem with $S
= S_\psi(G)$ and any given $\vc c$ and $r$ in time $r^{f(\psi,
  \gamma)}\cdot|V(G)|$ for some computable~$f$. 
\end{corollary}

While it is possible to prove Corollary~\ref{cor:sco_xp_cw}
along the same lines as used above, we avoid repeating the previous arguments and, instead,
apply the following technical tool.
\ifthenelse{\value{Accumulate} = 1}{}{%
This tool simply extends a folklore fact that a class of graphs is of bounded clique-width
if and only if it has an MSO$_1$ interpretation in the class of rooted trees.
}

\begin{lemma}\APXmark\label{lem:interpret_cw}
Let $G$ be a graph of clique-width $cw(G) = \gamma$
given along with a $\gamma$-expression~$\Gamma$ constructing~$G$, 
and let $\psi$ be an MSO$_1$ formula.
One can, in time $\Oh(|V(G)|+|\Gamma|+|\psi|)$, compute a tree $T$ and an MSO$_1$
formula $\varphi$ such that
$V(G)\subseteq V(T)$ and
$$\mbox{for every $X$ it is $T\models\varphi(X)$, \,iff\,
	$X\subseteq V(G)$ and $G\models\psi(X)$.}$$
\end{lemma}

\begin{accumulate}
\ifthenelse{\value{Accumulate} = 1}{\let\bfparg\bf}{\let\bfparg\relax}%
\paragraph{\bfparg Addition to Lemma~\ref{lem:interpret_cw}.}
Before giving a (short) proof of this lemma,
we need to formally introduce a simplified concept of MSO interpretability.
Let $\sigma$ and $\varrho$ be two relational vocabularies.
A~\emph{one-dimensional MSO interpretation}
of $\varrho$ in $\sigma$ is a tuple
$I = \big(\nu(x),\, \{\eta_R(\bar{x})\}_{R\in \varrho}\big)$ 
of $MSO[\sigma]$-formulas where
$\nu$ has one free element variable and the number of free element variables in
$\eta_R$ is equal to the arity of $R$ in $\varrho$.
\begin{itemize}
\item
To every $\sigma$-structure $A$ the interpretation $I$ assigns a $\varrho$-structure
$A^I$ with the domain $A^I = \{a~|~A \models \nu(a) \}$ and the relations $R^I =\{
{\bar{a}~|~A \models \eta_R(\bar{a})}\}$ for each $R \in \varrho$.  We say that
a class $\CC$ of $\varrho$-structures has an interpretation in a class
$\DD$ of $\sigma$-structures if there exists an interpretation $I$
such that for each $C \in \CC$ there exists $D \in \DD$ such
that $C \simeq D^I$, and for every  $D \in \DD$ the structure $D^I$
is isomorphic to a member of~$\CC$.
\item
The interpretation $I$ of $\varrho$ in $\sigma$ defines a translation of every
$MSO[\varrho]$-formula $\psi$ to an $MSO[\sigma]$-formula $\psi^I$ as follows:
\begin{itemize}
\item every $\exists x. \phi$ is replaced by $\exists x.(\nu(x) \land \phi^I)$,
\item every $\exists X. \phi$ is replaced by 
$\exists X.(\forall y (y \in X \rightarrow \nu(y)) \land \phi^I)$,
and
\item every occurrence of a $\sigma$-atom $R(\bar{x})$ is replaced by
the corresponding formula $\eta_R(\bar{x})$.
\end{itemize}  
\end{itemize}  
It is a folklore fact that for all $MSO[\varrho]$-formulas $\psi$ 
and all $\sigma$-structures $A$
$$ A \models \psi^I \Longleftrightarrow A^I \models \psi .$$

Let $G$ be a graph of clique-width $cw(G) = \gamma$
given along with a $\gamma$-expression~$\Gamma$ constructing~$G$, 
and let $\psi$ be an MSO$_1$ formula.
One can, in time $\Oh(|V(G)|+|\Gamma|+|\psi|)$, compute a tree $T$ and an MSO$_1$
formula $\varphi$ such that
$V(G)\subseteq V(T)$ and
$$\mbox{for every $X$ it is $T\models\varphi(X)$, \,iff\,
	$X\subseteq V(G)$ and $G\models\psi(X)$.}$$

\begin{proofof}{Lemma~\ref{lem:interpret_cw}}
Let $T_0$ be the parse tree of the given $\gamma$-expression~$\Gamma$
constructing~$G$.
Hence $T_0$ is a rooted tree such that the set of its leaves is $V(G)$.
It is well known that there is a one-dimensional MSO$_1$ interpretation
$I_1=\big(\nu_1(x),\eta_1(x,y)\big)$ of the graphs of clique-width $\leq\gamma$
in the class of colored rooted trees (such that the finite set of colors
depends only on~$\gamma$).
The used colors are in fact the vertex labels (in the leaves) and the
operator symbols (in the internal nodes) from Definition~\ref{def:cw},
and the interpretation $I_1$ is easy to construct for given $\gamma$.
See, e.g., \cite[Section~4.3]{DBLP:journals/eccc/Kreutzer09} for close details.
Consequently, $G\simeq T_0^{I_1}$.

To finish the proof, we just need to ``remove'' the colors from $T_0$
and ``forget'' the root (to make an ordinary uncolored tree $T$).
We first give the root of $T_0$ a new distinguished color $c_r$ (as a copy of its
original color).
Then the parent-child relation in $T_0$ can be easily interpreted based on
the unique path to a vertex colored $c_r$.
Let all the colors used in $I_1$ be $C=\{c_1,\dots,c_m\}$ where $m$ depends
on~$\gamma$ (including our distinguished root colors).
We construct a tree $T$ from $T_0$ by attaching $i$ new leaves to every
vertex of $T_0$ of color~$c_i\in C$.

We now straightforwardly define an MSO$_1$ interpretation
$I_2=\big(\nu_2(x),\eta_2(x,y),$
 $\lambda_2^i \mbox{ for $i\in[m]$}\big)$ of rooted $C$-colored trees in
the class of ordinary trees:
\begin{itemize}
\item $\nu_2(x)$ simply asserts that $x$ is not a leaf ($x$ has more than one neighbor),
\item $\eta_2(x,y)\equiv \prebox{edge}(x,y)$ (note that $T_0$ is an 
induced subgraph of~$T$), and
\item $\lambda_2^i(x)$ asserts that $x$ has precisely $i$ neighbors which
are leaves---this has a routine brute-force expression by existential
quantification of the $i$ leaf neighbors, coupled by non-existence of
$i+1$ leaf neighbors.
\end{itemize}
Clearly, $T^{I_2}\simeq T_0$ (including the colors of $T_0$).
We finally set $I=I_1\circ I_2$ (interpretability is a transitive concept)
and $\varphi=\psi^I$ and we are done.
\end{proofof}
\end{accumulate}

With Lemma~\ref{lem:interpret_cw} at hand, it is now easy to derive
Corollary~\ref{cor:sco_xp_cw} from previous Theorem~\ref{thm:sco_xp} applied
to the tree~$T$.

\begin{accumulate}
\begin{proofof}{Corollary~\ref{cor:sco_xp_cw}}
We invoke Lemma~\ref{lem:interpret_cw} to construct the tree $T$ and formula
$\varphi$, and the apply Theorem~\ref{thm:sco_xp} to them
(for a tree, $tw(T)=\tau=1$, but note that $\varphi$ now depends on both
$\psi$ and $\gamma$).
Since $S_\psi(G)$ is essentially identical to $S_\varphi(T)$, up to the coordinates
corresponding to $V(T)\setminus V(G)$ which are all zero by
Lemma~\ref{lem:interpret_cw},
the solution to the shifted problem with $S_\psi(G)$ is the same as the
computed solution to the shifted problem with $S_\varphi(T)$.
\end{proofof}
\end{accumulate}

Finally, we add a small remark regarding the input $G$ in
Corollary~\ref{cor:sco_xp_cw};
we are for simplicity assuming that $G$ comes along with its $\gamma$-expression 
since it is currently not known how to efficiently construct a $\gamma$-expression
for an input graph of fixed clique-width~$\gamma$.
Though, one may instead use the result of
\cite{DBLP:journals/siamcomp/HlinenyO08} which constructs in FPT a so-called
rank-decomposition of $G$ which can be used as an approximation of a
$\gamma$-expression for $G$ (with up to an exponential jump, but this does
not matter for a fixed parameter~$\gamma$ in theory).

\section{MSO-definable sets: W[1]-hardness}
\label{sec:mso_part_hardness}
\begin{accumulate}
\ifthenelse{\value{Accumulate} = 1}{
\section{Additions to Section~\ref{sec:mso_part_hardness}}
}{}\end{accumulate}

Recall that natural hard graph problems such as \textsc{Chromatic number} 
are instances of \textsc{MSO-partitioning} and so also instances of 
shifted combinatorial optimization.
While we have shown an \XP algorithm for SCO with MSO-definable sets
on graphs of bounded treewidth and clique-width
in Theorem~\ref{thm:sco_xp} and Corollary~\ref{cor:sco_xp_cw},
it is a natural question whether an \FPT algorithm could exist
for this problem, perhaps under a more restrictive width measure.

Here we give a strong negative answer to this question.
First, we note the result of Fomin et al~\cite{fgls09} proving
\W1-hardness of \textsc{Chromatic number} parameterized by the clique-width
of the input graph.
This immediately implies that an \FPT algorithm in
Corollary~\ref{cor:sco_xp_cw} would be very unlikely
(cf.~Section~\ref{sec:preliminaries}).
Although, \textsc{Chromatic number} is special in the sense that it is
solvable in \FPT when parameterized by the treewith of the input.
Here we prove that it is not the case of \textsc{MSO-partitioning} 
problems and SCO in general, even when considering restricted MSO$_1$ formulas
and shifted $\vc c$,
and parameterizing by a much more restrictive {\em treedepth} parameter.

\begin{definition}[Treedepth]
\label{def:tree-depth}
Let the {\em height} of a rooted tree or forest be the maximum root-to-leaf distance
in it.
The {\em closure} $cl(F)$ of a rooted forest $F$
is the graph obtained from $F$ by making every vertex adjacent to all of its
ancestors.
The {\em treedepth} $td(G)$ of a graph $G$ is one more than the minimum
height of a forest $F$ such that $G\subseteq cl(F)$.
\end{definition}

Note that always $td(G)\geq tw(G)+1$ since we can use the vertex sets of 
the root-to-leaf paths of the forest $F$ (from Definition~\ref{def:tree-depth}) 
in a proper order as the bags of a tree-decomposition of $G$.

\begin{theorem}\APXmark\label{thm:part-hardinstance}
There exists a graph FO formula $\varphi(X)$ with a free set variable $X$,
such that the instance of the \textsc{MSO-partitioning} problem given by $\varphi$,
is \W1-hard when parameterized by the treedepth of an input simple graph~$G$.
\\
Consequently, the shifted problem with $S_{\varphi}(G)$ is also \W1-hard
(for suitable $\vc c$) when parameterized by the treedepth of~$G$.
\end{theorem}

\ifthenelse{\value{Accumulate} = 1}{%
The way we approach Theorem~\ref{thm:part-hardinstance} is by a reduction from 
\W1-hardness of \textsc{Chromatic number} with respect to
clique-width~\cite{fgls09}, in which we exploit some special hidden properties
of that reduction, as extracted in~\cite{glo13}.
In a nutshell, the ``difficult cases'' of \cite{fgls09} can be interpreted
in a special way into labeled rooted trees of small height,
and here we further trade the labels (the number of which is the parameter) 
for increased height of a tree and certain additional edges belonging to the tree closure.

The full details and the proof are left for the Appendix.
}{%
We are going to prove Theorem~\ref{thm:part-hardinstance} by a reduction from 
\W1-hardness of \textsc{Chromatic number} with respect to clique-width~\cite{fgls09}.
As an intermediate step for our purpose, \cite{glo13} prove that the graphs constructed
for the reduction in~\cite{fgls09}, can be interpreted in a special way
(formal details to follow) into labeled rooted trees of height~$5$, 
where the parameter is the number of labels.
We, in turn, prove here that these labels can be traded for increased height of
a tree and certain additional edges belonging to the tree closure.
Consequently, the property of a set $X$ of vertices to be independent
in the original graph can now be expressed by a certain fixed formula $\varphi(X)$ 
(independent of the parameter) over a plain simple graph which is of bounded treedepth.
So the \textsc{MSO-partitioning} instance given by $\varphi$ is indeed
\W1-hard when parameterized by the treedepth.
}{}

\begin{accumulate}
We start with formulating the needed special reformulation of the aforementioned result
of Fomin et al~\cite{fgls09} on hardness of \textsc{Chromatic number}.

\begin{definition}[Tree-model~\cite{ganianetal12}]
\label{def:tree-model}
We say that a graph $G$ has a {\em tree-model of $m$ labels and depth $d$}
if there exists a rooted tree $T$ such that 
\begin{enumerate}
\item the set of leaves of $T$ is exactly $V(G)$,
\item the length of each root-to-leaf path in $T$ is exactly~$d$,
\item each leaf of $T$ is assigned one of $m$ labels,
\item\label{it:tree-model-edge}
and the existence of a $G$-edge between $u,v\in V(G)$ depends solely
on the labels of $u,v$ and the distance between $u,v$ in $T$.
\end{enumerate}
%
Let $\mathcal{T\!M}_m(d)$ denote the class of all graphs with a tree-model
of $m$ labels and depth $d$.
\end{definition}

\begin{theorem}[\cite{glo13}]\label{thm:coloringhard-tm}
The graphs constructed as the ``hard'' instances of \textsc{Chromatic number}
in~\cite{fgls09} belong to $\mathcal{T\!M}_m(5)$ where $m$ is the considered parameter.
Consequently, the \textsc{Chromatic number} problem considered on the classes 
$\mathcal{T\!M}_m(5)$ is \W1-hard when parameterized by~$m$.
\end{theorem}

\begin{proofof}{Theorem~\ref{thm:part-hardinstance}}
We use a reduction from the instance described in Theorem~\ref{thm:coloringhard-tm}.
Let $G\in\mathcal{T\!M}_m(5)$ and $r$ be an input of \textsc{Chromatic number},
i.e., the question is whether $G$ is $r$-colorable.
Let $T$ be a tree-model of $m$ labels and depth~$5$ 
(Definition~\ref{def:tree-model}) of~$G$.
We are going to construct a formula $\varphi$ and a graph $H$ of 
treedepth at most $5m+7$ such that, $V(G)\subseteq V(H)$ and
$X\subseteq V(G)$ is independent if and only if~$H\models\varphi(X)$.
Moreover, for $Y\subseteq V(H)$ such that
$Y\not\subseteq V(G)$, it must hold $H\models\varphi(Y)$
if and only if $Y=V(H)\setminus V(G)$.
Then, clearly, $(H,r+1)$ will be a \textsc{Yes} instance of
the \textsc{MSO-partitioning} problem given by $\varphi$ if, and only if,
$G$ is $r$-colorable.
This would be the desired reduction.

The rest of the proof is devoted to the construction of $\varphi$ and~$H$.
Let $M=\{1,\dots,m\}$ be the set of labels from
Definition~\ref{def:tree-model} and let $lab$ map $V(G)$, the set of leaves of
$T$, into~$M$. 
There exist graphs $L_1,\dots,L_5$ (self-loops allowed), each on the vertex set $M$,
such that the following holds for any $u,v\in V(G)$ by Definition~\ref{def:tree-model}:
$uv\in E(G)$ if and only if the least common ancestor of $u,v$ in
$T$ is at distance $i\leq 5$ from $u$ and $\{lab(u),lab(v)\}\in E(L_i)$.

A graph $H_1$ is constructed from $T$ as follows:
\begin{itemize}
\item All vertices and edges of $T$ are included in $H_1$
(the labels from $T$ are ignored),
the leaves of $T$ have no label in $H_1$ while all the non-leaf nodes get a new
label $\tau$ in $H_1$.
\item For every non-leaf node $x\in V(T)$ at distance $i\leq 5$ from the leaves of
$T$, a disjoint copy $L^x$ of the graph $L_i$ is created and added to $H_1$,
such that the vertices of $L^x$ receive the same (new) label~$\lambda$
and $x$ is made adjacent to all vertices of $L^x$.
\item Every leaf $z$ of $T$, for $i=1,\dots,5$, is connected 
by an edge in $H_1$ to the copy of the vertex $lab(z)$ in $L^x$, 
where $x$ is the ancestor of $z$ at distance $i$ from~$z$.
\end{itemize}

First of all, it is easy to see that $H_1$ is of treedepth at most $5m+6=1+5(m+1)$,
since $H_1$ is contained in the closure of a tree obtained from $T$ by
``splitting'' each non-leaf node $x$ to a path on $m+1$ vertices
forming the set $\{x\}\cup V(L^x)$ (which is of cardinality $m+1$).
Second, we observe that the graph $H_1$ encodes the edges of $G$ as follows;
(*)
for $u,v\in V(G)$, we find the least common ancestor $x$ of $u$ and $v$ among
the $\tau$-labeled vertices of $H_1$, and then we test whether there exist
$\lambda$-labeled neighbors $u',v'$ of $x$ (and so $u',v'\in V(L^x)$)
such that $uu',vv'\in E(H_1)$ and also~$u'v'\in E(H_1)$.

Assume for now that (*) the edges of $G$ are encoded in a
binary predicate $\gamma$, such that $uv\in E(G)$
$\iff$ $H_1\models\gamma(u,v)$.
With $\gamma$, we can easily define a desired formula $\varphi_1$
such that $H_1\models\varphi_1(X)$ if, and only if,
$X\subseteq V(G)$ is independent in~$G$ or $X=V(H_1)\setminus V(G)$.
It is
\begin{align*}
\varphi_1(X) \equiv&
	\left[\,
		\forall x\in X\big(\neg\tau(x)\wedge\neg\lambda(x)\big)
		\wedge \forall x,y\in X\big(x=y\vee \neg\gamma(x,y)\big)
	\,\right]
\\	&\vee
		\forall x\big((\tau(x)\vee\lambda(x))
		\longleftrightarrow x\in X\big)
.\end{align*}
Note that $\gamma$ does not depend on the original $m$ labels of~$T$.

The remaining two tasks are; to express $\gamma$ in FO over $H_1$,
and to ``get rid of'' possible self-loops and the labels $\tau,\lambda$ in~$H_1$
by transforming $\varphi_1$ over $H_1$ into equivalent $\varphi$ over simple
unlabeled $H$.
We finish these tasks as follows.

We recursively define $\alpha_0(x,y)\equiv(x=y)$ and, for $i=1,\dots,5$,
$\alpha_{i}(x,y)\equiv \tau(y)\wedge\exists z\big(
	\pedge(z,y)\wedge\alpha_{i-1}(x,z)\big)$.
The meaning of $\alpha_{i}(x,y)$ is that $y$ is an internal node of $T$ at
distance $i$ from $x$ (where $x$ will be a leaf of $T$ but this is not
enforced by $\alpha_i$).
The above encoding $(*)$ of the edges of $G$ into $H_1$ can now be
literally expressed as
\begin{align*}
\gamma(u,v) \equiv\>
	\exists x& \left[\, \tau(x) \wedge
		\left(\bigvee\nolimits_{i=1}^5
		\alpha_i(u,x)\wedge\alpha_i(v,x)\wedge
		\neg\exists x'\big(
			\alpha_{i-1}(u,x')\wedge\alpha_{i-1}(v,x')\big)\right)
	\right.
\\	&\wedge
	\exists u',v' \big(\lambda(u')\wedge\lambda(v')\wedge
		\pedge(x,u')\wedge\pedge(x,v')
\\	&\left.\vbox to 3ex{}\qquad~~~~\wedge
		\pedge(u,u')\wedge\pedge(v,v')\wedge\pedge(u',v')\big)
	\right]
,\end{align*}
which is an FO formula independent of $H_1$ and given~$T$.

Lastly, we observe that $H_1$ has no vertices of degree $1$.
We hence construct $H$ from $H_1$ by adding one new degree-$1$ neighbor to
every $\tau$-labeled vertex of $H_1$, adding two new degree-$1$ neighbors to  
every $\lambda$-labeled vertex of $H_1$ without self-loop,
adding three new degree-$1$ neighbors to 
every $\lambda$-labeled vertex of $H_1$ with self-loop,
and removing all the loops.
The resulting simple graph $H$ is of treedepth at most $5m+7$ (in fact, again
$\leq5m+6$), and one can identify the original vertices of $H_1$ as those
having degree $>1$ in~$H$.
The labels $\tau,\lambda$ and the self-loops of $H_1$ (as used in the formula
$\varphi_1$) can be routinely interpreted by FO formulas, e.g.,
$\tau(x)\equiv\neg\delta_1(x)\wedge\exists y
 \big(\delta_1(y)\wedge\pedge(x,y)\big)\wedge
 \forall y,y'\big[\big(\delta_1(y)\wedge\pedge(x,y)\wedge
  \delta_1(y')\wedge\pedge(x,y')\big)\to y=y'\big]$,
where $\delta_1(y)\equiv\forall z,z'\big[\big(
 \pedge(y,z)\wedge\pedge(y,z')\big)\to z=z'\big]$.
Such an interpretation defines desired $\varphi$ from~$\varphi_1$.
\end{proofof}
\end{accumulate}

\section{Conclusions and open problems}
\label{sec:conclusions}

We close with several open problems we consider
interesting and promising.

\smallskip
\textbf{Parameterizing by $r$. } It is interesting to consider
taking $r$ as a parameter. For example, Fluschnik et
al.~\cite{FluschnikKNS:2015} prove that the \textsc{Minimum Shared
  Edges} problem is \FPT parameterized by the number of paths. 
Omran et al.~\cite{OmranSZ:2013} prove that the \textsc{Minimum
  Vulnerability} problem is in \XP with the same parameter. Since both problems
are particular cases of the shifted problem, we ask whether the shifted
problem with $S$ being the set of $s-t$ paths of a (di)graph lies in \XP or is \NP-hard already for some constant $r$.

\smallskip
\textbf{Further uses of Lemma~\ref{lem:idp_separable}. }
For example, which interesting combinatorial sets $S$ can be represented
as $n$-fold integer programs~\cite{Onn:2010,DeLoeraHK:2013} such that the
corresponding polyhedra are decomposable?

\smallskip
\textbf{Approximation. } The \textsc{Minimum Vulnerability}
problem has also been studied from the perspective of
approximation algorithms~\cite{OmranSZ:2013}. What can be said about
the approximation of the shifted problem?

\smallskip
\textbf{Going beyond $0/1$. } The results in
Section~\ref{sec:explicit} are the only known ones in which $S$ does not have
to be $0/1$. What can be said about the shifted problem with such sets $S$
that are not given explicitly, e.g., when $S$ is given by a totally
unimodular system?

\bibliographystyle{plain}
\bibliography{shifted}

\ifaccumulating{%
\newpage\appendix\sloppy
\section{Appendix}
\let\epr=\endproofof
\accuprint
}

\end{document}